%% file: paper.tex
\documentclass[11pt]{article}
\input{preamble.tex}

\title{On the complexity of hazard-free circuits}
\author[1]{Christian Ikenmeyer}
\author[2]{Balagopal Komarath}
\author[1]{Christoph Lenzen}
\author[2,3]{Vladimir Lysikov}
\author[4]{Andrey Mokhov}
\author[5]{Karteek Sreenivasaiah\footnote{This work was done while this author was at Saarland University.\\Author email addresses: cikenmey$@$mpi-inf.mpg.de, bkomarath$@$cs.uni-saarland.de, clenzen$@$mpi-inf.mpg.de, vlysikov$@$cs.uni-saarland.de, andrey.mokhov$@$ncl.ac.uk, karteek$@$iith.ac.in}}

\affil[1]{Max Planck Institute for Informatics, Saarland Informatics Campus}
\affil[2]{Saarland University, Saarland Informatics Campus}
\affil[3]{Cluster of Excellence MMCI, Saarland Informatics Campus}
\affil[4]{School of Engineering, Newcastle University}
\affil[5]{Indian Institute of Technology Hyderabad}

\date{\today}

\begin{document}
\maketitle
\thispagestyle{empty}
\begin{abstract}
\input{abstract.tex}
\end{abstract}

\bigskip

{\footnotesize \noindent\textbf{2012 ACM Computing Classification System:} Theory of computation -- Computational complexity and cryptography -- Circuit complexity}

\bigskip


\bigskip

{\footnotesize \noindent\textbf{Acknowledgments:} We thank Arseniy Alekseyev, Karl Bringmann, Ulan Degenbaev, Stefan Friedrichs, and Matthias F\"ugger for helpful discussions. Moreover, we thank Attila Kinali for his help with Japanese references. This project has received funding from the European Research Council (ERC) under the European Union's Horizon 2020 research and innovation programme (grant agreement n\degree\,716562). Andrey Mokhov's research was supported by EPSRC grant POETS (EP/N031768/1).}

\clearpage
\setcounter{page}{1}

\input{intro.tex}
\input{relwork.tex}

\input{definitions.tex}

\input{lowerbounds.tex}
\input{upperbounds.tex}
\input{detection.tex}
\input{future.tex}
\appendix
\input{bases.tex}

\input{majtree.tex}

\bibliographystyle{alpha}
\bibliography{paper}

\end{document}

%% file: preamble.tex
\usepackage{authblk}
\usepackage{graphicx,latexsym,amsmath,amsthm,amssymb,color,verbatim,mathrsfs,bbm,tikz,float,gensymb}
\usepackage{CJKutf8}
\usepackage[margin=1in]{geometry}
\usepackage{csquotes}
\usepackage{xcolor}
\usepackage[colorlinks=true,linkcolor=blue,citecolor=blue,filecolor=blue,urlcolor=blue]{hyperref}
\usepackage{soul}
\usepackage{tikz}
\usepackage{indentfirst}
\usepackage{mathtools}
\usepackage[font=small]{subcaption}
\usepackage[font=small]{caption}

\swapnumbers
\numberwithin{equation}{section}
\newtheorem{theorem}[equation]{Theorem}
\newtheorem*{theorem*}{Theorem}
\newtheorem{corollary}[equation]{Corollary}
\newtheorem{lemma}[equation]{Lemma}

\newtheorem{proposition}[equation]{Proposition}

\newtheorem{definition}[equation]{Definition}

\theoremstyle{definition}

\newcommand{\IB}{\mathbb{B}}

\newcommand{\AND}{\textup{\texttt{and}}}
\newcommand{\XOR}{\textup{\texttt{xor}}}
\newcommand{\OR}{\textup{\texttt{or}}}
\newcommand{\NOT}{\textup{\texttt{not}}}
\newcommand{\MUX}{\texttt{MUX}}

\newcommand{\fixedhazard}{{\mathsf{FixedHazard}}}
\newcommand{\hazard}{{\mathsf{Hazard}}}
\newcommand{\SimpleHazard}{{\mathsf{OneBitHazard}}}
\newcommand{\ZeroCircuitFixedHazard}{{\mathsf{UnsatFixedHazard}}}
\newcommand{\mfu}{{\mathfrak{u}}}
\newcommand{\IT}{{\mathbb{T}}}
\newcommand{\diff}{\mathop{}\!\mathrm{d}}

\DeclarePairedDelimiter\floor{\lfloor}{\rfloor}

\newcommand{\NP}{\mathrm{NP}}

\newcommand{\poly}{\operatorname{poly}}

\newcommand{\GF}{\operatorname{GF}}

%% file: abstract.tex
The problem of constructing hazard-free Boolean circuits dates back to the 1940s and is an important problem in circuit design.
Our main lower-bound result unconditionally shows the existence of
functions whose circuit complexity is polynomially bounded while every
hazard-free implementation is provably of exponential size.
Previous lower bounds on the hazard-free complexity were only valid for depth 2 circuits.
The same proof method yields that every subcubic implementation of Boolean matrix multiplication must have hazards.

These results follow from a crucial structural insight:
Hazard-free complexity is a natural generalization of monotone complexity to all (not necessarily monotone) Boolean functions.
Thus, we can apply known monotone complexity lower bounds to find lower bounds on the hazard-free complexity.
We also lift these methods from the monotone setting to prove exponential hazard-free complexity lower bounds for non-monotone functions.

As our main upper-bound result we show how to efficiently convert a Boolean circuit into a bounded-bit hazard-free circuit with only a polynomially large blow-up in the number of gates.
Previously, the best known method yielded exponentially large circuits in the worst case, so our algorithm gives an exponential improvement.

As a side result we establish the NP-completeness of several hazard detection problems.

%% file: intro.tex
\section{Introduction}\label{sec:intro}

We study the problem of \emph{hazards} in Boolean circuits. This problem
naturally occurs in digital circuit design, specifically in the implementation
of circuits in hardware (e.g.\ \cite{HUF:57, CAL:58}), but is also closely
related to questions in logic (e.g.\ \cite{KLE:52, KOR:66, MAL:14}) and
cybersecurity (\cite{TWMMCS:09, WOITSMK:12}). Objects are called differently in
the different fields; for presentational simplicity, we use the parlance of
hardware circuits throughout the paper.

A Boolean circuit is a circuit that uses $\AND$-, $\OR$-, and $\NOT$-gates, in
the traditional sense of \cite[problem~MS17]{GJ:79}, i.e., $\AND$ and $\OR$ have
fan-in two. The standard approach to studying hardware implementations of
Boolean circuits is to use the digital abstraction, in which voltages on wires
and at gates are interpreted as either logical $0$ or $1$. More generally, this
approach is suitable for any system in which there is a guarantee that the
inputs to the circuit and the outputs of the gates of the circuit can be reliably
interpreted in this way (i.e., be identified as the Boolean value matching the
gate's truth table).

\subsection{Kleene Logic and Hazards}
Several independent works~(\cite{GOT:48}, \cite{YR:64} and references therein) observed
that Kleene's classical three-valued \emph{strong logic of indeterminacy}
$K_3$~\cite[\S64]{KLE:52} captures the issues arising from non-digital inputs.
The idea is simple and intuitive. The two-valued Boolean logic is extended by a
third value~$\mfu$ representing any unknown, uncertain, undefined, transitioning, or
otherwise non-binary value. We call both Boolean values \emph{stable}, while
$\mfu$ is called \emph{unstable}. The behavior of a Boolean gate is then
extended as follows. Let $\IB:=\{0,1\}$ and $\IT:=\{0,\mfu,1\}$. Given a string
$x\in \IT^k$, a \emph{resolution} $y \in \IB^k$ of $x$ is defined as a string
that is obtained by replacing each occurrence of $\mfu$ in $x$ by either $0$ or $1$.
If a $k$-ary gate (with one output) is subjected to inputs $x\in \IT^k$,
it outputs $b\in \IB$ iff it outputs $b$ for \emph{all} resolutions $y \in \IB^k$ of $x$,
otherwise it outputs $\mfu$.
In other words, the gate outputs a Boolean value $b$,
if and only if its output does not actually depend on the unstable inputs. This
results in the following extended specifications of $\AND$, $\OR$, and $\NOT$ gates:
\begin{center}
\begin{tabular}{c|ccc}
${\NOT}$ & 0 & $\mfu$ & 1 \\
\hline
& 1&$\mfu$&0
\end{tabular}
\hspace{2cm}
\begin{tabular}{c|ccc}
${\AND}$ & 0 & $\mfu$ & 1 \\
\hline
 0 & 0 & 0 & 0\\
 $\mfu$ & 0 & $\mfu$ & $\mfu$\\
 1 & 0 & $\mfu$ & 1
\end{tabular}
\hspace{2cm}
\begin{tabular}{c|ccc}
${\OR}$  & 0 & $\mfu$ & 1 \\
\hline
 0 & 0 & $\mfu$ & 1\\
 $\mfu$ & $\mfu$ & $\mfu$ & 1\\
 1 & 1 & 1 & 1
\end{tabular}
\end{center}
By induction over the circuit
structure, a circuit $C$ with $n$ input gates now computes a function $C\colon \IT^n
\to \IT$.

Unfortunately, in some cases, the circuit might behave in an undesirable way. Consider a
multiplexer circuit (MUX), which for Boolean inputs $x,y,s\in \IB$ outputs $x$ if
$s=0$ and $y$ if $s=1$. A straightforward circuit implementation is shown in Figure~\ref{fig:mux}.
Despite the fact that
$\MUX(1,1,0)=\MUX(1,1,1)=1$, one can verify that in Figure~\ref{fig:mux},
$\MUX(1,1,\mfu)=\mfu$. Such behaviour is called a hazard:
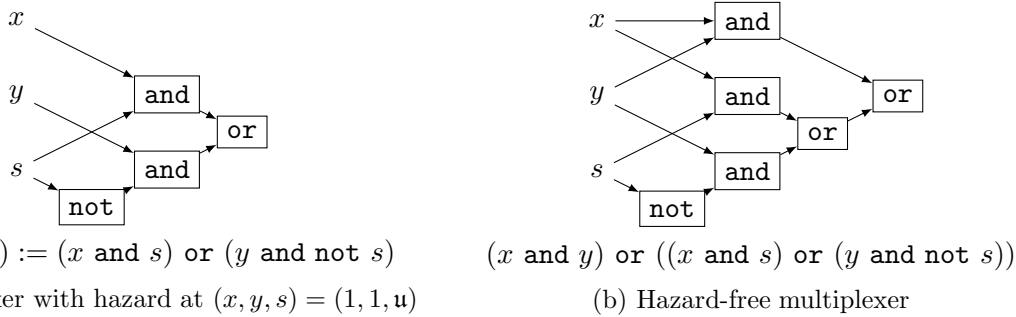
\begin{figure}[t]
\centering
\begin{subfigure}{0.49\textwidth}
\centering
\begin{tikzpicture}
\node (x) at (-1,2) {$x$};
\node (y) at (-1,1) {$y$};
\node (s) at (-1,0) {$s$};
\node[draw,rectangle] (and2) at (1,0) {\AND};
\node[draw,rectangle] (and1) at (1,1) {\AND};
\node[draw,rectangle] (not) at (0,-0.5) {\NOT};
\node[draw,rectangle] (or) at (2,0.5) {\OR};
\draw[-latex] (x) -- (and1);
\draw[-latex] (s) -- (and1);
\draw[-latex] (s) -- (not);
\draw[-latex] (y) -- (and2);
\draw[-latex] (not) -- (and2);
\draw[-latex] (and1) -- (or);
\draw[-latex] (and2) -- (or);
\end{tikzpicture}\\
$\MUX(x,y,s) := (x \ \AND\ s) \ \OR \ (y \ \AND\ \NOT\ s)$
\caption{Multiplexer with hazard at $(x,y,s)=(1,1,\mfu)$}
\label{fig:mux}
\end{subfigure}
\begin{subfigure}{0.49\textwidth}
\centering
\begin{tikzpicture}
\node (x) at (-1,2) {$x$};
\node (y) at (-1,1) {$y$};
\node (s) at (-1,0) {$s$};
\node[draw,rectangle] (and2) at (1,0) {\AND};
\node[draw,rectangle] (and1) at (1,1) {\AND};
\node[draw,rectangle] (not) at (0,-0.5) {\NOT};
\node[draw,rectangle] (or) at (2,0.5) {\OR};
\node[draw,rectangle] (andnew) at (1,2) {\AND};
\node[draw,rectangle] (ornewlast) at (3,1) {\OR};
\draw[-latex] (x) -- (and1);
\draw[-latex] (s) -- (and1);
\draw[-latex] (s) -- (not);
\draw[-latex] (y) -- (and2);
\draw[-latex] (not) -- (and2);
\draw[-latex] (and1) -- (or);
\draw[-latex] (and2) -- (or);
\draw[-latex] (x) -- (andnew);
\draw[-latex] (y) -- (andnew);
\draw[-latex] (andnew) -- (ornewlast);
\draw[-latex] (or) -- (ornewlast);
\end{tikzpicture}\\
$(x \ \AND \ y) \ \OR \ ((x \ \AND\ s) \ \OR \ (y \ \AND\ \NOT\ s))$
\caption{Hazard-free multiplexer}
\label{fig:cmux}
\end{subfigure}
\caption{Two circuits that implement the same Boolean multiplexer function. One has a hazard, the other one is hazard-free.}
\end{figure}
\begin{definition}[Hazard]\label{def:hazard}
We say that a circuit $C$ on $n$ inputs \emph{has a hazard 
at} $x \in \IT^n$ iff $C(x)=\mfu$ and there is a Boolean value $b \in \IB$ such that for all
resolutions $y$ of $x$ we have $C(y)=b$. If $C$ has no hazard, it is called
\emph{hazard-free.}
\end{definition}
The name \emph{hazard-free} has different meanings in the literature. Our definition is taken from \cite{DDT:78}. %[Sec.~3.3.4.2]
In Figure~\ref{fig:cmux} we see a hazard-free circuit for the multiplexer function.
Note that this circuit uses more gates than the one in Figure~\ref{fig:mux}.
The problem of detecting hazards and constructing circuits that are
hazard-free started a large body of literature, see
Section~\ref{sec:relwork}. The question whether hazards can be avoided
in principle was settled by Huffman.
\begin{theorem}[{\cite{HUF:57}}]\label{thm:huffman}
Every Boolean function has a hazard-free circuit computing it.
\end{theorem}
He immediately noted that avoiding hazards is potentially
expensive~\cite[p.~54]{HUF:57}:
\begin{quote}
``In this example at least, the elimination of hazards required a substantial
increase in the number of contacts.''
\end{quote}
Indeed, his result is derived using a clause construction based on the prime
implicants of the considered function, which can be exponentially many, see
e.g.\ \cite{CM:78}. There has been no significant progress on the complexity of
hazard-free circuits since Huffmann's work. Accordingly, the main question we
study in this paper is:
\begin{center}
What is the additional cost of making a circuit hazard-free?
\end{center}

\subsection{Our Contribution}
%###
\paragraph*{Unconditional lower bounds.}
%###
Our first main result is that monotone circuit lower bounds directly
yield lower bounds on hazard-free circuits. A circuit is \emph{monotone} if it
only uses $\AND$-gates and $\OR$-gates, but does not use any $\NOT$-gates. For a Boolean
function $f$, denote (i) by $L(f)$ its Boolean complexity, i.e., the size of a
smallest circuit computing $f$, (ii) by $L_\mfu(f)$ its \emph{hazard-free
complexity,} i.e., the size of a smallest hazard-free circuit computing $f$,
and (iii), if $f$ is monotone, by $L_+(f)$ its monotone circuit complexity,
i.e., the size of a smallest monotone circuit computing~$f$. We show that
$L_\mfu$ properly extends $L_+$ to the domain of all Boolean functions.
\begin{theorem}\label{thm:monmetaeq}
If $f$ is monotone, then $L_\mfu(f)=L_+(f)$.
\end{theorem}
We consider this connection particularly striking, because hazard-free circuits
are highly desirable in practical applications, whereas monotone circuits may
seem like a theoretical curiosity with little immediate applicability. Moreover,
to our surprise the construction underlying Theorem~\ref{thm:monmetaeq} yields a
circuit computing a new directional derivative
that we call the \emph{hazard derivative}\footnote{Interestingly,
this is closely related to, but \emph{not} identical to, the Boolean directional
derivative defined in e.g.\ \cite[Def.3]{dRSdlVG:12}, which has applications in
cryptography. To the best of our knowledge, the hazard derivative has not appeared in the literature so far.}
of the function at $x=0$ in direction of $y$, which equals the
function itself if it is monotone (and not constant $1$). We consider this
observation to be of independent interest, as it provides additional insight
into the structure of hazard-free circuits.

We get the following (non-exhaustive) list of immediate corollaries that
highlight the importance of Theorem~\ref{thm:monmetaeq}.

\begin{corollary}[using monotone lower bound from \cite{RAZBOROV-1985-MONOTONE}]\label{cor:permanent}
Define the Boolean permanent function $f_n\colon \IB^{n^2} \to \IB$ as
$$f(x_{11}, \dots, x_{nn}) = \bigvee_{\sigma \in S_n} \bigwedge_{i = 1}^n
x_{i\sigma(i)}.$$ We have $L(f_n) = O(n^5)$ and $L_{\mfu}(f_n) \geq
2^{\Omega(\log^2 n)}$.
\end{corollary}
\begin{corollary}[using monotone lower bound from \cite{TAR:88}]\label{cor:tardos}
  There exists a family of functions $f_n\colon \IB^{n^2} \to \IB$
  such that $L(f_n) = \poly(n)$ and $L_{\mfu}(f_n) \geq 2^{c
    n^{1/3 - o(1)}}$ for a constant $c>0$.
\end{corollary}
In particular, there is an exponential separation between $L$ and $L_\mfu$,
where the difference does not originate from an artificial exclusion of $\NOT$
gates, but rather from the requirement to avoid hazards. We even obtain
separation results for \emph{non-monotone} functions!
\begin{corollary}\label{cor:determinant}
  Let $\det_n\colon \IB^{n^2} \to \IB$ be the determinant over the field with $2$ elements, that is,
  \[ {\det}_n(x_{11}, \dots, x_{nn}) = \bigoplus_{\sigma \in S_n} \prod_{i = 1}^n x_{i \sigma(i)}.\]
  We have $L(\det_n) = \poly(n)$ and $L_{\mfu}(\det_n) \geq 2^{\Omega(\log^2
  n)}$.
\end{corollary}
Another corollary of Theorem~\ref{thm:monmetaeq} separates circuits of linear size from their hazard-free counterparts.
\begin{corollary}[using monotone lower bound from \cite{AB:87}]\label{cor:alon_boppana}
  There exists a family of functions $f_n \colon \IB^{N} \to \IB$ such that $L(f_n) = O(N)$ but $L_{\mfu}(f_n) \geq 2^{c N^{1/4 - o(1)}}$ for some $c > 0$, where the number of input variables of $f_n$ is $N = 4^n + \lfloor \frac{2^{n/2}}{4\sqrt{n}}\rfloor$
\end{corollary}

As a final example, we state a weaker, but still substantial separation result for Boolean matrix multiplication.

\begin{corollary}[using monotone lower bound from \cite{Pat:75, MG:76}, see also the earlier \cite{Pra:74}]\label{cor:matrix}
Let $f:\IB^{n\times n}\times \IB^{n\times n} \to \IB^{n\times n}$ be the Boolean
matrix multiplication map, i.e., $f(X,Y)=Z$ with $z_{i,j} = \bigvee_{k=1}^n
x_{i,k}\wedge y_{k,j}$. Every circuit computing $f$ with fewer than $2 n^3-n^2$
gates has a hazard. In particular, every circuit that implements Strassen's
algorithm \cite{str:69} or any of its later improvements (see e.g.\
\cite{leg:14}) has a hazard.
\end{corollary}
Since our methods are based on relabeling circuits only, analogous translations
can be performed for statements about other circuit complexity measures, for
example, the separation result for the circuit depth from
\cite{RAZ-WIGDERSON-1992-MONOTONE-DEPTH}. The previously best lower bounds on
the size of hazard-free circuits are restricted to depth 2 circuits (with
unbounded fan-in and not counting input negations), see
Section~\ref{sec:relwork}.

%###
\paragraph*{Parametrized upper bound.}
%###
These hardness results imply that we cannot hope for a general construction of a small hazard-free circuit for $f$ even if $L(f)$ is small.
However, the task becomes easier when restricting to
hazards with a limited number of unstable input bits.
\begin{definition}[$k$-bit hazard]
For a natural number $k$, a circuit $C$ on $n$ inputs has a \emph{$k$-bit}
hazard at $x \in \IT^n$, iff $C$ has a hazard at $x$ and $\mfu$ appears at
most $k$ times in $x$.
\end{definition}
Such a restriction on the number of unstable input bits has been considered in many papers (see e.g.\
\cite{YR:64,ZKK:79,UNG:95,WOITSMK:12}),
but the state-of-the-art in terms of asymptotic complexity has not improved
since Huffman's initial construction~\cite{HUF:57}, which is of size
exponential in $n$, see the discussion of \cite{TY:12,TYK:14} in
\cite[Sec.~``Speculative Computing'']{Fri:17}.
We present a construction with blow-up exponential in $k$,
but polynomial in $n$. In particular, if $k$ is constant and $L(f_n)\in
\poly(n)$, this is an exponential improvement.
\begin{corollary}\label{cor:kHazardFree}
Let $C$ be a circuit with $n$ inputs, $|C|$ gates and depth $D$.
Then there is a circuit with $\left(\frac{ne}{k}\right)^{2k}(|C|+6)+O(n^{2.71 k})$
gates and depth $D + 8k + O(k \log n)$
that computes the same function and has no $k$-bit hazards.
\end{corollary}

%###
\paragraph*{Further results.}
%###
We round off the presentation by a number of further results. First, to
further support the claim that the theory of hazards in circuits is natural, we
prove that it is independent of the set of gates (\AND, \OR, \NOT), as long as
the set of gates is functionally complete and contains a constant function, see Corollary~\ref{cor:functionallycomplete}.
Second, it appears unlikely that much more than logarithmically many unstable bits can be handled with only a polynomial circuit size blow-up.

\begin{theorem}\label{thm:conditional-lower-bound}
Fix a monotonously weakly increasing sequence of natural numbers $k_n$ with $\log n \leq k_n$ and set $j_n := k_n/\log n$.
If Boolean circuits deciding $j_n$-CLIQUE on graphs with $n$ vertices require a circuit size of at least $n^{\Omega(j_n)}$, then
there exists a function $f_n : \IB^{n^2+k_n}\to\IB$
with $L(f_n)=\poly(n)$
for which circuits without $k_n$-bit hazards require $2^{\Omega(k_n)}$ many gates to compute.
\end{theorem}
In particular, if $k_n = \omega(\log n)$ is only slightly superlogarithmic, then
Theorem~\ref{thm:conditional-lower-bound} provides a function where the circuit size blow-up is superpolynomial if we insist on having no $k_n$-bit hazards.
In this case $j_n$ is slightly superconstant,
which means that ``Boolean circuits deciding $j_n$-CLIQUE require size at least $n^{\Omega(j_n)}$''
is a consequence of a nonuniform version of the exponential time hypothesis (see \cite{LMS:11}),
i.e., smaller circuits would be a major algorithmic breakthrough.

We remark that, although it has not been done before, deriving conditional lower
bounds such as Theorem~\ref{thm:conditional-lower-bound} is rather straightforward. In contrast,
Theorem~\ref{thm:monmetaeq} yields \emph{unconditional} lower bounds.

Finally, determining whether a circuit has a hazard is NP-complete,
even for 1-bit hazards (Theorem~\ref{thm:detecthazardnpcomplete}).
This matches the fact that the best algorithms for these tasks have
exponential running time \cite{EIC:65}.  Interestingly, this also
means that if $\NP \neq \mathrm{coNP}$, given a circuit there exists
no polynomial-time verifiable certificate of size polynomial in the
size of the circuit to prove that the circuit is hazard-free, or even
free of 1-bit hazards.

%% file: relwork.tex
\section{Related work}\label{sec:relwork}
Multi-valued logic is a very old topic and several three-valued logic
definitions exist. In 1938 Kleene defined his strong logic of indeterminacy
\cite[p.~153]{KLE:38}, see also his later textbook \cite[\S64]{KLE:52}.
It can be readily defined by setting $\mfu=\frac 1 2$, $\NOT\ x := 1-x$, $x \
\AND \ y := \min(x,y)$, and $x \ \OR \ y := \max(x,y)$, as it is commonly done
in fuzzy logic \cite{PR:79, Roj:96}. This happens to model the behavior of
physical Boolean gates and can be used to formally define hazards. This was
first realized by Goto in \cite[p.~128]{GOT:49}, which is the first paper that
contains a hazard-free implementation of the multiplexer, see
\cite[Fig.~7$\cdot$5]{GOT:49}. The third truth value in circuits was mentioned
one year earlier in \cite{GOT:48}.
As far as we know, this early Japanese work was unnoticed in the Western world at first.
The first structural results on hazards appeared in a seminal paper by Huffman \cite{HUF:57},
who proved that every Boolean function has a hazard-free circuit.
This is also the first paper that observes the apparent circuit size blow-up that occurs when insisting on a hazard-free implementation of a function.
Huffman mainly focused on 1-bit hazards, but notes that his methods carry over to general hazards.
Interestingly, our Corollary~\ref{cor:kHazardFree} shows that for 1-bit hazards the circuit size blow-up is polynomially bounded,
while for general hazards we get the strong separation of Corollary~\ref{cor:tardos}.

The importance of hazard-free circuits is already highlighted for example in the classical textbook \cite{CAL:58}.
Three-valued logic for circuits was introduced by Yoeli and Rinon in \cite{YR:64}.
In 1965, Eichelberger published the influential paper \cite{EIC:65}, which shows
how to use three-valued logic to detect hazards in exponential time.
This paper also contains the first lower bound on hazard-free depth 2 circuits: A hazard-free $\AND$-$\OR$ circuit with negations at the inputs
must have at least as many gates as its function has prime implicants,
which can be an exponentially large number, see e.g.\ \cite{CM:78}.
Later work on lower bounds was also only concerned with depth 2 circuits, for example \cite{ND:92}.

Mukaidono \cite{MUK:72} was the first to formally define a partial order of definedness, see also \cite{MUK:83A,MUK:83B},
where it is shown that a ternary function is computable by a circuit iff it is monotone under this partial order.
In 1981 Marino \cite{MAR:81} used a continuity argument to show (in a more general context) that specific ternary functions cannot be implemented, for example
there is no circuit that implements the detector function $f(\mfu)=1$, $f(0)=f(1)=0$.

Nowadays the theory of three-valued logic and hazards can be found for example in the textbook \cite{BS:95}.
A fairly recent survey on multi-valued logic and hazards is given in \cite{BEI:01}.

Recent work models clocked circuits~\cite{FFL:16}.
Applying the standard technique of ``unrolling'' a clocked circuit into a combinational circuit, one sees that the computational power of clocked and unclocked circuits is the same.
Moreover, lower and upper bounds translate between the models as expected;
using $r$ rounds of computation changes circuit size by a factor of at most $r$.
However, \cite{FFL:16} also models a special type of registers, masking registers, that have the property that if they output $\mfu$ when being read in clock cycle $r$, they output a stable value in all subsequent rounds (until written to again).
With these registers, each round of computation enables computing strictly more (ternary) functions.
Interestingly, adding masking registers also breaks the relation between hazard-free and monotone complexity: \cite{FFL:16} presents a transformation that trades a factor $O(k)$ blow-up in circuit size for eliminating $k$-bit hazards.
In particular, choosing $k=n$, a linear blow-up suffices to construct a hazard-free circuit out of an arbitrary hazardous implementation of a Boolean function.

Seemingly unrelated, in 2009 a cybersecurity paper \cite{TWMMCS:09} was
published that studies information flow on the Boolean gate level.
The logic of the information flow happens to be Kleene's logic and thus results transfer in both directions.
In particular (using different nomenclature) they design a circuit (see \cite[Fig.~2]{TWMMCS:09})
that computes the Boolean derivative, very similar to our construction in Proposition~\ref{prop:transformation-1}.
In the 2012 follow-up paper \cite{WOITSMK:12} the construction of this circuit is monotone (see \cite[Fig.~1]{WOITSMK:12})
which is a key property that we use in our main structural correspondence result Theorem~\ref{thm:monmetaeq}.

There is an abundance of monotone circuit lower bounds that all translate to hazard-free complexity lower bounds, for example \cite{RAZBOROV-1985-MONOTONE, AG:87, Yao:89, RAZ-WIGDERSON-1992-MONOTONE-DEPTH}
and references in \cite{GS:92} for general problems,
but also \cite{Weg:82} and references therein for explicit problems,
\cite{Pra:74, Pat:75, MG:76} for matrix multiplication and \cite{Blu:85} for the Boolean convolution map.
This last reference also implies that any implementation of the Fast Fourier Transform to solve Boolean convolution must have hazards.

On a very high level, some parts of our upper bounds construction in Section~\ref{sec:hazard-free} are reminiscent to \cite[Prop.~6.5]{NW:96} or \cite{Gol:11}.

%% file: definitions.tex
\section{Definitions}\label{sec:defs}
We study functions $F \colon \IT^n \to \IT$ that can be implemented by circuits.
The Boolean analogue is just the set of \emph{all} Boolean functions.
In our setting this is more subtle.
First of all, if a circuit gets a Boolean input, then by the definition of the gates it also outputs a Boolean value.
Thus every function that is computed by circuits \emph{preserves stable values}, i.e., yields a Boolean value on a Boolean input.
Now we equip $\IT$ with a partial order $\preceq$ such that $\mfu$ is the
least element and $0$ and $1$ are incomparable elements greater than
$\mfu$, see \cite{MUK:72}. We extend this order to $\IT^n$ in the usual way. 
For tuples $x, y\in \IT^n$ the statement $x \preceq y$ means that
$y$ is obtained from $x$ by replacing some unstable values with stable ones.
Since the gates $\AND$, $\OR$, $\NOT$ are monotone with respect to $\preceq$,
every function $F$ computed by a circuit must be monotone with respect to $\preceq$.
It turns out that these two properties capture precisely what can be computed:
\begin{proposition}[{\cite[Thm.~3]{MUK:72}}]
A function $F \colon \IT^n \to \IT$ can be computed by a circuit iff $F$ preserves stable values and is monotone with respect to $\preceq$.
\end{proposition}
A function $F\colon \IT^n \to \IT$ that preserves stable values and is monotone with respect to $\preceq$ shall be called a \emph{natural function}.
A function $F \colon \IT^n \to \IT$ is called an \emph{extension} of a Boolean
function $f \colon \IB^n \to \IB$ if the restriction~$F|_{\IB^n}$ coincides with $f$.

Observe that any natural extension $F$ of a Boolean function $f$ must satisfy the following.
If $y$ and $y'$ are resolutions of $x$ (in particular $x\preceq y$ and $x\preceq y'$) such that $F(y)\neq F(y')$, it must hold that $F(y)=0$ and $F(y')=1$ (or vice versa), due to preservation of stable values.
By $\preceq$-monotonicity, this necessitates that $F(x)=\mfu$, the only value ``smaller'' than both $0$ and $1$.
Thus, one cannot hope for a stable output of a circuit if $x$ has two resolutions with different outputs.
In contrast, if all resolutions of $x$ produce the same output, we can require a stable output for $x$, i.e., that a circuit computing $F$ is hazard-free.
\begin{definition}
For a Boolean function $f\colon \IB^n \to \IB$, define its \emph{hazard-free extension} $\bar{f}\colon \IT^n \to \IT$ as follows:
\begin{align*}  
\bar{f}(x) =
    \begin{cases}
      0, &\text{if $f(y) = 0$ for all resolutions $y$ of $x$,}\\
      1, &\text{if $f(y) = 1$ for all resolutions $y$ of $x$,}\\
      \mfu, &\text{otherwise}.
    \end{cases}
\end{align*}
\end{definition}
Hazard-free extensions are natural functions and are exactly those functions that are computed by hazard-free circuits, as can be seen for example by
Theorem~\ref{thm:huffman}.
Equivalently, $\bar{f}$ is the unique extension of $f$ that is monotone and maximal with respect to $\preceq$.

We remark that later on we will also use the usual order $\leq$ on $\IB$ and $\IB^n$.
We stress that the term \emph{monotone Boolean function} refers to functions $\IB^n \to \IB$ monotone with respect to $\leq$.

%% file: lowerbounds.tex
\section{Lower bounds on the size of hazard-free circuits}

In this section, we prove that $L_{\mfu}(f)= L_+(f)$ for monotone functions $f$, from which Corollaries~\ref{cor:permanent} to~\ref{cor:matrix} follow.
Our first step is to show that $L_{\mfu}(f)\leq L_+(f)$, which is straightforward.

\subsection{Conditional lower bound}
\label{sec:conditional-lower-bound}

In this section we prove Theorem~\ref{thm:conditional-lower-bound}, which is a direct consequence of the following proposition
and noting that $n^{\Omega(k_n/\log n)} = 2^{\Omega(k_n)}$.

\begin{proposition}\label{prop:doubling}
Fix a monotonously weakly increasing sequence of natural numbers $j_n$ with $j_n \leq n$.
There is a function $f_n: \IB^{n^2+j_n \log n}\to\IB$
with $L(f_n) = \poly(n)$ and the following property:
if $f_n$ can be computed by circuits of size $L_n$ that are free of $(j_n \log n)$-bit hazards,
then there are Boolean circuits of size $2 L_n$ that decide $j_n$-CLIQUE.
\end{proposition}
\begin{proof}
The function $f_n$ gets as input the adjacency matrix of a graph $G$ on $n$ vertices and a list $\ell$ of $j_n$ vertex indices, each encoded in binary with $\log n$ many bits:
\[
f_n(G,\ell) = \begin{cases}
                      1 & \text{ if  $\ell$ encodes a list of $j_n$ vertices that form a $j_n$-clique in $G$}, \\
                      0 & \text{ otherwise}.
                      \end{cases}
\]
Clearly $L(f_n) = \poly(n)$.
Let $C$ compute $f_n$ and have no $(j_n\log n)$-hazards.
By the definition of $(j_n\log n)$-hazards, it follows that
$C(G,\mfu^{j_n\log n}) \neq 0$ iff $G$ contains a $j_n$-clique.
From $C$ we construct a circuit $C'$ that decides $j_n$-CLIQUE as follows.
We double each gate and each wire. Additionally, after each doubled $\NOT$-gate we twist the two wires so that this $\NOT$ construction sends $(0,1)$ to $(0,1)$ instead of to $(1,0)$.
Stable inputs to $C$ are doubled, whereas the input $\mfu$ is encoded as the Boolean pair $(0,1)$.
It is easy to see that the resulting circuit simulates $C$.
Our circuit $C'$ should have $n^2$ inputs and should satisfy $C'(G)=1$ iff $C(G,\mfu^{j_n\log n})\neq 0$,
thus we fix the $j_n\log n$ rightmost input pairs to constants $(0,1)$ to obtain $C'$.
From the two output gates, we treat the right output gate as the output of $C'$, while dismissing the left output gate.
\end{proof}

\subsection{Monotone circuits are hazard-free}
\begin{lemma}\label{lem:monotonearehazardfree}
Monotone circuits are hazard-free. In particular, for monotone Boolean functions $f$ we have $L_{\mfu}(f)\leq L_+(f)$.
\end{lemma}
\begin{proof}
We prove the claim by induction over the number of computation gates in the circuit.
Trivially, a monotone circuit without computation gates is hazard-free, as it merely forwards some input to the output.
For the induction step, let $C$ be a monotone circuit computing a function $F\colon \IT^n\to \IT$
such that the gate computing the output of $C$ receives as inputs the outputs of two hazard-free monotone subcircuits $C_1$ and $C_2$.
We denote by $F_1$ and $F_2$ the natural functions computed by $C_1$ and $C_2$, respectively.
The gate computing the output of $C$ can be an $\AND$- or an $\OR$-gate and we will treat both cases in parallel.
Let $x\in \IT^n$ be arbitrary with the property that $F(y)=1$ for all resolutions $y$ of~$x$.
Denote by $y_0$ the resolution of $x$ in which all $\mfu$'s are replaced by $0$.
The fact that $F(y_0)=1$ implies that $F_1(y_0)=F_2(y_0)=1$ ($F_1(y_0)=1$ or $F_2(y_0)=1$).
By monotonicity of $F_1$ and $F_2$, this extends from $y_0$ to all resolutions $y$ of $x$, because $y\geq y_0$ and thus $F(y)\geq F(y_0)=1$.
Since $C_1$ and $C_2$ are hazard-free by the induction hypothesis,
we have $F_1(x)=F_2(x)=1$ ($F_1(x)=1$ or $F_2(x)=1$).
As basic gates are hazard-free, we conclude that $F(x)=1$.

The case that $F(y)=0$ for all resolutions $y$ of some $x\in \IT^n$ is analogous, where $y_0$ is replaced by $y_1$, the resolution of $x$ in which all $\mfu$'s are replaced by $1$.
\end{proof}

The following sections show a much deeper relationship between monotone and hazard-free circuits.
A key concept is the derivative, which we will discuss next.

\subsection{Derivatives of natural functions}

Let $F \colon \IT^n \to \IT$ be a natural function and $x \in \IB^n$ be a stable input.
If $\tilde{x} \preceq x$, that is, if $\tilde{x}$ is obtained from $x$ by replacing stable bits by $\mfu$, then $F(\tilde{x}) \preceq F(x)$.
This means that there are two possibilities for $F(\tilde{x})$ --- either $F(\tilde{x}) = F(x)$ or $F(\tilde{x}) = \mfu$.

We can encode in one Boolean function the information about how the value of $F$ changes from $F(x)$ to $\mfu$ when the bits of the input change from stable to unstable.
It is reminiscent of the idea of the derivative in analysis or the Boolean derivative, which also show how the value of the function changes when the input changes.
To make the connection more apparent, we introduce a notation for replacing stable bits by unstable ones:
if $x, y \in \IB^n$, then $x + \mfu y$ denotes the tuple that is obtained from $x$ by changing the values to $\mfu$ in all positions in which $y$ has a 1, and keeping the other values unchanged.
Formally,
\[ \tilde{x} = x + \mfu y \quad \Leftrightarrow \quad \tilde{x}_i = \begin{cases} x_i, & \text{if $y_i = 0$,} \\ \mfu, &\text{if $y_i = 1$.} \end{cases} \]
This notation is consistent with interpreting the addition and multiplication on $\IT$ as the hazard-free extensions of the usual addition modulo $2$ and multiplication on $\IB$ ($\XOR$ and $\AND$).

Any tuple $\tilde{x} \preceq x$ can be presented as $x + \mfu y$ for some $y \in \IB^n$.
As we have seen, $F(x + \mfu y)$ is either $F(x)$ or $\mfu$.
This condition can also be written as $F(x + \mfu y) = F(x) + \mfu \Delta$ for some $\Delta \in \IB$.

\begin{definition}
  Let $F \colon \IT^n \to \IT$ be a natural function.
  The \emph{hazard derivative} (or just \emph{derivative} for short) of $F$ is the Boolean function $\diff F \colon \IB^{n} \times \IB^{n} \to \IB$ such that
  \begin{equation}
  \label{eq:derivative-definition}
    F(x + \mfu y) = F(x) + \mfu \cdot \diff F(x; y).
  \end{equation}
  In other words,
  \[ \diff F(x; y) = \begin{cases}
    0, &\text{if $F(x + \mfu y) = F(x)$,} \\
    1, &\text{if $F(x + \mfu y) = \mfu$.}
  \end{cases}
  \]
For a Boolean function $f$ we use the shorthand notation $\diff f := \diff \overline f$.
\end{definition}

Consider for example the disjunction $\OR$.
The values of $(x_1 + \mfu y_1) \OR (x_2 + \mfu y_2)$ are as follows:
\begin{center}
\begin{tabular}{c|cccc}
  ${\OR}$            & $0 + \mfu \cdot 0$ & $0 + \mfu \cdot 1$ & $1 + \mfu \cdot 0$ & $1 + \mfu \cdot 1$ \\
  \hline
  $0 + \mfu \cdot 0$ & $0$                & $\mfu$             & $1$                & $\mfu$             \\
  $0 + \mfu \cdot 1$ & $\mfu$             & $\mfu$             & $1$                & $\mfu$             \\
  $1 + \mfu \cdot 0$ & $1$                & $1$                & $1$                & $1$                \\
  $1 + \mfu \cdot 1$ & $\mfu$             & $\mfu$             & $1$                & $\mfu$
\end{tabular}
\end{center}
Thus, 
\begin{subequations}
  \label{eq:derivative-basis}
  \begin{align}
    \diff \OR(x_1, x_2; y_1, y_2) &= \neg x_1 y_2 \vee \neg x_2 y_1 \vee y_1 y_2. \\
  \intertext{Similarly, we find}
    \diff \NOT(x; y) &= y, \\
    \diff \AND(x_1, x_2; y_1, y_2) &= x_1 y_2 \vee x_2 y_1 \vee y_1 y_2, \\
    \diff \XOR(x_1, x_2; y_1, y_2) &= y_1 \vee y_2.
  \end{align}
\end{subequations}

Caveat:
Since natural functions $F$ are exactly those ternary functions defined by circuits,
we can obtain $\diff F$ from the ternary evaluations of any circuit computing $F$.
For Boolean functions $f$ it is more natural to think of $\diff f$ as a property of the function $f$, because the correspondence to circuits is not as close:
we can obtain $\diff f$ from the ternary evaluations of any \emph{hazard-free} circuit computing $f$ on Boolean inputs.

In general, we can find the derivative of a Boolean function as follows:
\begin{lemma}
\label{lem:derivative-closed}
  For $f\colon \IB^n \to \IB$, we have $\diff{f}(x; y) = \bigvee_{z \leq y} [f(x) + f(x + z)]$.
  In particular, if $f(0) = 0$, then $\diff{f}(0; y) = \bigvee_{z \leq y} f(z)$.
\end{lemma}
\begin{proof}
  Resolutions of $x + \mfu y$ coincide with $x$ at positions where $y$ has a~$0$ and have arbitrary stable bits at positions where $y$ has a~$1$.
  Therefore, each resolution of $x + \mfu y$ can be presented as $x + z$ for some $z$ such that $z_i = 0$ whenever $y_i = 0$, that is, $z \leq y$.
  Hence, the set of all resolutions of $x + \mfu y$ is $S(x + \mfu y):=\{x + z \mid z \leq y \}$.

  The derivative $\diff f(x; y) = 1$ if and only if $\bar{f}(x + \mfu y) = \mfu$.
  By definition of hazard-freeness, this happens when $f$ takes both values $0$ and $1$ on $S(x + \mfu y)$,
  in other words, when the $f(x+z)\neq f(x)$ for some $z \in S(x + \mfu y)$.
  The disjunction $\bigvee_{z \leq y} [f(z) + f(x + z)]$ represents exactly this statement.
\end{proof}

As a corollary, we obtain a surprisingly close relation between monotone Boolean functions and their derivatives.
For a natural function $F$ and any fixed $x \in \IB^n$, let $\diff F(x;.)$ denote the Boolean function that maps $y\in\IB^n$ to $\diff F(x;y)$,
and define the shorthand $\diff f(x;.) := \diff \overline{f}(x;.)$ for a Boolean function $f$.

\begin{corollary}
\label{cor:derivative-monotone}
Suppose that $f\colon \IB^n \to \IB$ is monotone with $f(0) = 0$. Then $\diff f(0,.) = f$.
\end{corollary}

\begin{lemma}
\label{lem:derivative-monotone}
  For natural $F \colon \IT^n \to \IT$ and fixed $x \in \IB^n$, $\diff F(x; .)$ is a monotone Boolean function.
\end{lemma}
\begin{proof}
Note that the expression $x + \mfu y$ is antimonotone in $y$: if $y_1 \geq y_2$, i.e., $y_1$ is obtained from $y_2$ by replacing $0$s with $1$s, then $x + \mfu y_1$ is obtained from $x + \mfu y_2$ by replacing more stable bits of $x$ with $\mfu$, so $x + \mfu y_1 \preceq x + \mfu y_2$.
Thus, if $y_1 \geq y_2$, $F$ being natural yields that \[F(x) + \mfu \diff F(x; y_1) = F(x + \mfu y_1) \preceq F(x + \mfu y_2) = F(x) + \mfu \diff F(x; y_2), \] so $\diff F(x; y_1) \geq \diff F(x; y_2)$.
\end{proof}

We can also define derivatives for vector functions $F \colon \IT^n \to \IT^m$, $F(x) = (F_1(x), \dots, F_m(x))$ with natural components $F_1, \dots, F_m$ as $\diff F(x; y) = (\diff F_1(x; y), \dots, \diff F_m(x; y))$.
Note that the equation \eqref{eq:derivative-definition} still holds and uniquely defines the derivative for vector functions.

The following statement is the analogue of the chain rule in analysis.

\begin{lemma}[Chain rule]
\label{lem:chain-rule}
  Let $F \colon \IT^n \to \IT^m$ and $G \colon \IT^m \to \IT^l$ be natural functions and $H(x) = G(F(x))$.
  Then \[\diff H(x; y) = \diff G(F(x); \diff F(x; y)).\]
\end{lemma}
\begin{proof}
  Use equation \eqref{eq:derivative-definition}.
\begin{eqnarray*}
H(x + \mfu y) &=& G(F(x + \mfu y)) = G(F(x) + \mfu \diff F(x; y)) = G(F(x)) + \mfu \diff G(F(x); \diff F(x; y)) \\
&=& H(x) + \mfu \diff G(F(x); \diff F(x; y)),
\end{eqnarray*}
and the claim follows with another application of \eqref{eq:derivative-definition}.
\end{proof}

\subsection{Using monotone circuits to compute derivatives}
In this section we show how to efficiently compute derivatives by transforming circuits to monotone circuits.
Our main tool is the chain rule (Lemma \ref{lem:chain-rule}).

For a circuit $C$ and a gate $\beta$ of $C$, let $C_\beta$ denote the natural function computed at the gate $\beta$.
\begin{proposition}
\label{prop:transformation-1}
From a circuit $C$ we can construct a circuit $C'$
by independently replacing each gate $\beta$
  on $t$ inputs $\alpha_1,\ldots,\alpha_t$ ($0 \leq t \leq 2$)
  by a subcircuit on $2t$ inputs $\alpha_1,\ldots,\alpha_t,\alpha'_1,\ldots,\alpha'_t$ and two output gates $\beta,\beta'$ (the wiring between these subcircuits in $C'$ is the same as the wiring between the gates in $C$, but in $C'$ we have two parallel wires for each wire in $C$)
  such that $C'_\beta(x,y) = C_\beta(x)$
and $C'_{\beta'}(x,y) = \diff C_\beta(x;y)$ for Boolean inputs $x,y\in \IB^n$.
\end{proposition}
\begin{proof}
  To construct $C'$, we extend $C$ with new gates.
  For each gate $\beta$ in $C$, we add a new gate $\beta'$.
  If $\beta$ is an input gate $x_i$, then $\beta'$ is the input gate $y_i$.
  If $\beta$ is a constant gate, then $\beta'$ is the constant-$0$ gate.

  The most interesting case is when $\beta$ is a gate implementing a function $\varphi \in \{\AND,\OR,\NOT\}$ with incoming edges from gates $\alpha_1, \dots, \alpha_t$
  (in our definition of the circuit, the arity $t$ is $1$ or $2$, but the construction works without modification in the general case).
  In this case, we add to $\beta$ a subcircuit which takes $\alpha_1, \dots, \alpha_t$ and their counterparts $\alpha'_1, \dots, \alpha'_t$ as inputs
  and $\beta'$ as its output gate,
  which computes $C'_{\beta'}(x,y)=\diff \varphi(C'_{\alpha_1}(x,y), \dots, C'_{\alpha_t}(x,y); C'_{\alpha'_1}(x,y), \dots, C'_{\alpha'_t}(x,y))$.
  For the sake of concreteness, for the gate types $\NOT$, $\AND$, $\OR$ according to \eqref{eq:derivative-basis} this construction is depicted in Figure~\ref{fig:gatereplacements}.
  
  \begin{figure}[h]
    \centering
    \begin{tikzpicture}[>=latex]
      \begin{scope}[shift={(-3,0)}]
        \node (a) at (-1,0) {$\alpha$};
        \node (b) at (1,0) {$\beta$};
        \node[draw,rectangle] (not) at (0,0) {$\NOT$};
        \draw[->] (a)--(not);
        \draw[->] (not)--(b);
      \end{scope}
      \node at (0,0) {$\Rightarrow$};
      \begin{scope}[shift={(3.5,0)}]
        \node (a) at (-1,1) {$\alpha$};
        \node (b) at (1,1) {$\beta$};
        \node (a') at (-1,-1) {$\alpha'$};
        \node (b') at (1,-1) {$\beta'$};
        \node[draw,rectangle] (not) at (0,1) {$\NOT$};
        \draw[->] (a)--(not);
        \draw[->] (not)--(b);
        \draw[->] (a')--(b');
      \end{scope}
      \begin{scope}[shift={(-3,-4)}]
        \node (a1) at (-1,-0.5) {$\alpha_1$};
        \node (a2) at (-1,0.5) {$\alpha_2$};
        \node (b) at (1,0) {$\beta$};
        \node[draw,rectangle] (and) at (0,0) {$\AND$};
        \draw[->] (a1)--(and);
        \draw[->] (a2)--(and);
        \draw[->] (and)--(b);
      \end{scope}
      \node at (0,-4) {$\Rightarrow$};
      \begin{scope}[shift={(4.5,-4)}]
        \node (a1) at (-2.5,1-0.5) {$\alpha_1$};
        \node (a2) at (-2.5,1+0.5) {$\alpha_2$};
        \node (b) at (2,1) {$\beta$};
        \node (a1') at (-2.5,-1-0.5) {$\alpha'_1$};
        \node (a2') at (-2.5,-1+0.5) {$\alpha'_2$};
        \node (b') at (2,-1) {$\beta'$};
        \node[draw,rectangle] (and) at (0,1) {$\AND$};
        \node[draw,rectangle] (and1) at (-1,-1+1) {$\AND$};
        \node[draw,rectangle] (and2) at (-1,-1) {$\AND$};
        \node[draw,rectangle] (and3) at (-1,-1-1) {$\AND$};
        \node[draw,rectangle] (or) at (0,-1) {$\OR$};
        \node[draw,rectangle] (orN) at (1,-1) {$\OR$};
        \draw[->] (a1)--(and);
        \draw[->] (a2)--(and);
        \draw[->] (and)--(b);
        \draw[->] (a2)--(and1);
        \draw[->] (a1')--(and1);
        \draw[->] (a1)--(and2);
        \draw[->] (a2')--(and2);
        \draw[->] (a1')--(and3);
        \draw[->] (a2')--(and3);
        \draw[->] (and1)--(orN);
        \draw[->] (and2)--(or);
        \draw[->] (and3)--(or);
        \draw[->] (or)--(orN);
        \draw[->] (orN)--(b');
      \end{scope}
      \begin{scope}[shift={(-3,-9)}]
        \node (a1) at (-1,-0.5) {$\alpha_1$};
        \node (a2) at (-1,0.5) {$\alpha_2$};
        \node (b) at (1,0) {$\beta$};
        \node[draw,rectangle] (or) at (0,0) {$\OR$};
        \draw[->] (a1)--(or);
        \draw[->] (a2)--(or);
        \draw[->] (or)--(b);
      \end{scope}
      \node at (0,-9) {$\Rightarrow$};
      \begin{scope}[shift={(5,-9)}]
        \node (a1) at (-3.5,1-0.5) {$\alpha_1$};
        \node (a2) at (-3.5,1+0.5) {$\alpha_2$};
        \node (b) at (2,1) {$\beta$};
        \node (a1') at (-3.5,-1-0.5) {$\alpha'_1$};
        \node (a2') at (-3.5,-1+0.5) {$\alpha'_2$};
        \node (b') at (2,-1) {$\beta'$};
        \node[draw,rectangle] (or) at (-1,1) {$\OR$};
        \node[draw,rectangle] (not1) at (-2.5,-1+0.7) {$\NOT$};
        \node[draw,rectangle] (not2) at (-2.3,-1+1.3) {$\NOT$};
        \node[draw,rectangle] (and1) at (-1,-1+1) {$\AND$};
        \node[draw,rectangle] (and2) at (-1,-1) {$\AND$};
        \node[draw,rectangle] (and3) at (-1,-1-1) {$\AND$};
        \node[draw,rectangle] (or3) at (0,-1) {$\OR$};
        \node[draw,rectangle] (orN) at (1,-1) {$\OR$};
        \draw[->] (a1)--(or);
        \draw[->] (a2)--(or);
        \draw[->] (or)--(b);
        \draw[->] (a1)--(not1);
        \draw[->] (a2)--(not2);
        \draw[->] (not2)--(and1);
        \draw[->] (a1')--(and1);
        \draw[->] (not1)--(and2);
        \draw[->] (a2')--(and2);
        \draw[->] (a1')--(and3);
        \draw[->] (a2')--(and3);
        \draw[->] (and1)--(orN);
        \draw[->] (and2)--(or3);
        \draw[->] (and3)--(or3);
        \draw[->] (or3)--(orN);
        \draw[->] (orN)--(b');
      \end{scope}
    \end{tikzpicture}
    \caption{Gates in $C$ get replaced by subcircuits in the construction of $C'$.}
    \label{fig:gatereplacements}
  \end{figure}
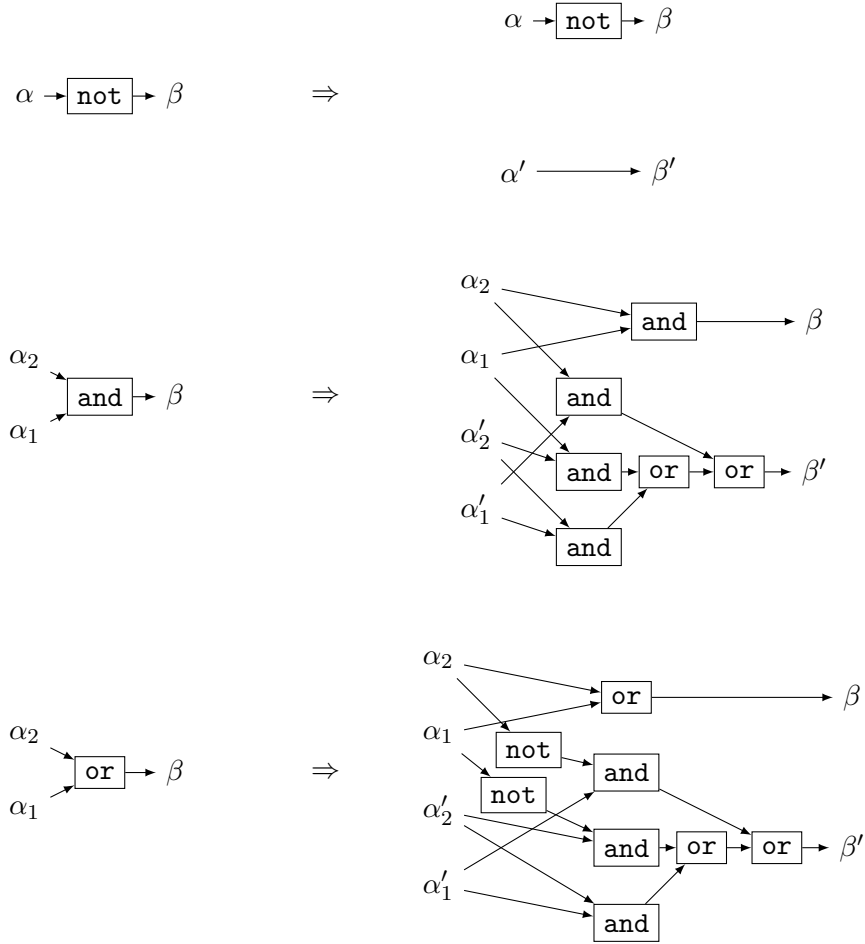

  Clearly $C'_\beta(x,y)=C_\beta(x)$.
  By induction on the structure of the circuit, we now prove that
  $C'_{\beta'}(x,y)=\diff C_{\beta}(x; y)$.
  In the base case, if $\beta$ is an input or constant gate, the claim is obvious.
  If $\beta$ is a gate of type $\varphi\in\{\AND,\OR,\NOT\}$ with incoming edges from $\alpha_1, \dots, \alpha_t$, then
  \[ C_{\beta}(x) = \varphi(C_{\alpha_1}(x), \dots, C_{\alpha_t}(x)).\]
  By the chain rule,
  \[ \diff C_{\beta}(x; y) = \diff \varphi(C_{\alpha_1}(x), \dots, C_{\alpha_t}(x); \diff C_{\alpha_1} (x; y), \dots, \diff C_{\alpha_t}(x; y)).\]
  By the induction hypothesis, $(\alpha'_1, \dots, \alpha'_t)=(\diff C_{\alpha_1} (x; y), \dots, \diff C_{\alpha_t}(x; y))$, thus the induction step succeeds by construction of $C'_{\beta'}$.
\end{proof}

Note that this construction can be seen as simulation of the behavior of the circuit $C$ on the input $x + \mfu y$:
the value computed at the gate $\beta$ on this input is $C_{\beta}(x) + \mfu \diff C_{\beta}(x; y)$,
and in $C'$ the gates $\beta$ and $\beta'$ compute the two parts of this expression separately.

By fixing the first half of the input bits in $C'$ we now establish the link to monotone complexity.
In the following theorem the case $x=0$ will be of particular interest.

\begin{theorem}
\label{thm:transformation-2}
For $f : \IB^n \to \IB$ and fixed $x \in \IB^n$, it holds that $L_{+}(\diff f(x,.)) \leq L_\mfu(f)$.
\end{theorem}
\begin{proof}
Let $C$ be a hazard-free circuit for $f$ of minimal size and let $x \in \IB^n$ be fixed.
We start by constructing the circuit $C'$ from Proposition~\ref{prop:transformation-1} and for each gate in $C$ we remember the corresponding subcircuit in $C'$.
For each subcircuit we call the gates $\alpha_i$ the \emph{primary inputs} and the $\alpha_i'$ the \emph{secondary inputs}.
From $C'$ we now construct a monotone circuit $C^x$ on $n$ inputs that computes $\diff f(x;.)$ as follows.
We fix the leftmost $n$ input bits $x \in \IB^n$ in $C'$.
This assigns a Boolean value $C'_{\alpha}(x) = C_{\alpha}(x)$ to each primary input $\alpha$ in each constructed subcircuit.
Each constructed subcircuit's secondary output $\beta'$ now computes some Boolean function in the secondary inputs~$\alpha'_i$.
If the values at the secondary inputs are $u_1 = C'_{\alpha'_1}(x, y), \dots, u_t = C'_{\alpha'_t}(x, y)$,
then the value at the secondary output is $\psi(u_1, \dots, u_t) = \diff \varphi(C_{\alpha_1}(x), \dots, C_{\alpha_t}(x); u_1, \dots, u_t)$.
Lemma~\ref{lem:derivative-monotone} implies that $\psi$ is monotone (which can alternatively be seen directly from Figure~\ref{fig:gatereplacements}, where fixing all primary inputs makes all $\NOT$ gates superfluous).
However, the only monotone functions on at most two input bits are the identity (on one input), $\AND$, $\OR$, and the constants.
Thus, we can replace each subcircuit in $C'$ by (at most) one monotone gate, yielding the desired monotone circuit $C^x$ that has at most as many gates as $C$ and outputs $\diff f(x;.) = \diff \bar{f}(x;.)=\diff C(x;.)=C'(.)$, where the second equality holds because $C$ is hazard-free.
\end{proof}

We now use this construction to prove Theorem~\ref{thm:monmetaeq}.

\begin{proof}[Proof of Theorem~\ref{thm:monmetaeq}]
The claim is trivial for the constant 1 function. Note that this is the only case of a monotone function that has $f(0)\neq 0$.
Hence assume that $f$ is monotone with $f(0)=0$.
By Lemma~\ref{lem:monotonearehazardfree}, we have that $L_{\mfu}(f) \leq L_{+}(f)$.
The other direction can be seen via
  $L_{+}(f) \stackrel{\text{Cor.}\,\ref{cor:derivative-monotone}}{=} L_{+}(\diff f(0,.)) \stackrel{\text{Thm.}\,\ref{thm:transformation-2}}{\leq} L_{\mfu}(f)$.
\end{proof}

Theorem~\ref{thm:monmetaeq} shows that the hazard-free complexity $L_{\mfu}$ can be seen as an extension of monotone complexity $L_+$ to general Boolean functions.
Thus, known results about the gap between general and monotone complexity transfer directly to hazard-free complexity.

\subsection{Unconditional lower bounds}
Corollaries~\ref{cor:permanent}, \ref{cor:tardos}, and~\ref{cor:matrix} are immediate applications of Theorem~\ref{thm:monmetaeq}.
Interestingly, however, we can also derive results on \emph{non-monotone} functions, which is illustrated by Corollary~\ref{cor:determinant}.
\begin{proof}[Proof of Corollary~\ref{cor:determinant}]
  The fact that the determinant can be computed efficiently is well known.

  Consider the derivative $\diff \det_n(0; y) = \bigvee_{z \leq y} \det_n(z)$ (Lemma \ref{lem:derivative-closed}).
  If there exists a permutation $\pi \in S_n$ such that all $y_{i \pi(i)}$ are $1$, then, replacing all the other entries with $0$ we get a matrix $z \leq y$ with $\det_n(z) = 1$, and $\diff \det_n(0; y) = 1$.
  If there is no such permutation, then all the summands in the definition of $\det_n(y)$ are $0$, and this is also true for all matrices $z \leq y$. In this case, $\diff \det_n(0; y) = 0$.
  Combining both cases, we get that $\diff \det_n(0; .)$ equals the Boolean permanent function $f_n$ from Corollary~\ref{cor:permanent}.
  The lower bound then follows from \cite{RAZBOROV-1985-MONOTONE} and Theorem~\ref{thm:transformation-2} (as in Corollary~\ref{cor:permanent}).
\end{proof}

We can combine this technique with the ideas from the proof of Theorem~\ref{thm:conditional-lower-bound} to show even stronger separation results, exhibiting a family of functions for which the complexity of Boolean circuits is \emph{linear,} yet the complexity of hazard-free circuits grows almost as fast as in Corollary~\ref{cor:tardos}.

\begin{lemma}
\label{lem:hazardfree-verifier-monotone}
  Let $f \colon \IB^n \to \IB$ be a monotone Boolean function with $f(0) = 0$ and $g \colon \IB^{n + m} \to \IB$ be a function
  such that $f(x) = 1$ iff $g(x, y) = 1$ for some $y \in \IB^m$.
  Then $L_{+}(f) \leq L_{\mfu}(g)$.
\end{lemma}
\begin{proof}
  Using Lemma \ref{lem:derivative-closed}, we obtain
  \[ \diff g(0, 0; x, 1) = \bigvee_{(z, t) \leq (x, 1)} g(z, t) = \bigvee_{z \leq x} \bigvee_t g(z, t) = \bigvee_{z \leq x} f(z) = f(z),\]
  which means that the circuit for $f$ can be obtained from the circuit for $\diff g(0;.)$ by substituting $1$ for some inputs. The statement then follows from Theorem \ref{thm:transformation-2}.
\end{proof}
\begin{proof}[Proof of Corollary~\ref{cor:alon_boppana}]
  We use the $\NP$-complete family $\mathrm{POLY}(q, s)$ from the paper of Alon and Boppana \cite{AB:87}.
  Let $\GF(q)$ denote a finite field with $q$ elements.
  We encode subsets $E \subset \GF(q)^2$ using $q^2$ Boolean variables in a straightforward way.
  The function $\mathrm{POLY}(q, s)$ maps $E \subset \GF(q)^2$ to $1$ iff there exists a polynomial $p$ of degree at most $s$ over $\GF(q)$ such that $(a, p(a)) \in E$ for every $a \in \GF(q)$.

  Alon and Boppana proved that for $s \leq \frac12 \sqrt{\frac{q}{\ln q}}$ the monotone complexity of this function is at least $q^{c s}$ for some constant $c$.
  For simplicity, we choose $q = 2^n$ and $s = \lfloor \frac14 \sqrt{\frac{q}{\log q}}\rfloor = \lfloor \frac{2^{n/2}}{4 \sqrt{n}} \rfloor$.
  In this case, $L_{+}(\mathrm{POLY}(q, s)) \geq 2^{c q^{1/2} \sqrt{\log q}}$.

  We define $f_n$ as the verifier for this instance of $\mathrm{POLY}$.
  The function $f_n$ takes $q^2 + sq = O(q^2)$ variables.
  The first $q^2$ inputs encode a subset $E \subset \GF(q)^2$,
  and the second $sn$ inputs encode coefficients of the polynomial $p$ of degree at most $s$ over $\GF(q)$, each coefficient using $n$ bits.
  The value $f_n(E, p) = 1$ iff $(a, p(a)) \in E$ for all $a \in \GF(q)$.
  To implement the function $f_n$, for each element $a \in \GF(q)$ we compute the value $p(a)$ using finite field arithmetic.
  Each such computation requires $O(s n^2)$ gates.
  Then we use $p(a)$ as a selector in a multiplexer to compute the value indicating whether $(a, p(a))$ is contained in $E$,
  choosing it from all the bits of the input $E$ corresponding to pairs of form $(a, b)$.
  This multiplexer requires additional $O(q)$ gates for each element $a \in \GF(q)$.
  The result is the conjunction of the computed values for all $a \in \GF(q)$.
  The total size of the circuit $O(q^2 + qsn^2 + q)$ is linear in the size of the input.

  The lower bound on the hazard-free complexity follows from the Alon-Boppana lower bound and Lemma \ref{lem:hazardfree-verifier-monotone}.
\end{proof}

%% file: upperbounds.tex
\section{Constructing \texorpdfstring{$k$}{k}-bit hazard-free circuits}\label{sec:hazard-free}
In this section we prove Corollary~\ref{cor:kHazardFree}.

For a collection $T$ of subsets of $[n]$, denote by $L_T(f)$ the minimum size of a circuit whose outputs coincide with $\bar{f}$ whenever the set of input positions with unstable bits is a subset of a set in the collection $T$.
Thus, $\preceq$-monotonicity of natural functions implies that $L(f)=L_{\emptyset}(f)\leq L_T(f)\leq L_{\{[n]\}}(f)=L_{\mfu}(f)$.
Excluding $k$-bit hazards therefore means that we consider $T=\binom{[n]}{k}$, i.e., $T$ contains all subsets of $[n]$ with exactly $k$ elements. The minimum circuit depth $D_T(f)$ is defined analogously.

As the base case of our construction, we construct circuits handling only fixed positions for the (up to) $k$ unstable bits, i.e., $T=\{S\}$ for some $S\in \binom{[n]}{k}$.
This is straightforward with an approach very similar to speculative computing~\cite{TY:12,TYK:14}.

We take $2^k$ copies of a circuit computing $f$. In the $i$th copy $(0 \leq i < 2^k)$ we fix the inputs in $S$ to the binary representation of $i$.
Now we use a hazard-free multiplexer to select one of these $2^k$ outputs, where the original input bits from $S$ are used as the select bits.
A hazard-free $k$-bit multiplexer of size $O(2^k)$ can be derived from the $1$-bit construction given in Figure~\ref{fig:cmux}.
\begin{lemma}\label{lem:k_bit_mux}
A $k$-bit multiplexer $\MUX_k$ receives inputs $x\in \IB^{2^k}$ and $s\in \IB^k$. It interprets $s$ as index from $[2^k]$ and outputs $x_s$. There is a hazard-free circuit for $\MUX_k$ of size $6(2^k-1)$
and depth $4k$.
\end{lemma}
\begin{proof}
A hazard-free $\MUX_1$ of size $6$ and depth $4$ is given in Figure~\ref{fig:cmux};
its correctness is verified by a simple case analysis.
From a hazard-free $\MUX_k$ and the hazard-free $\MUX_1$ we construct a hazard-free $\MUX_{k+1}$ circuit $C$ as follows:
\begin{eqnarray*}
\MUX_{k+1}(x_1,\ldots,x_{2^{k+1}};s_1,\ldots,s_{k+1}) = \MUX_1(&&\MUX_k(x_1,\ldots,x_{2^k};s_1,\ldots,s_k), \\ && \MUX_k(x_{2^k+1},\ldots,x_{2^{k+1}};s_1,\ldots,s_k);\quad s_{k+1}).
\end{eqnarray*}
One can readily verify that the resulting Boolean function is $\MUX_k$, and it has the desired circuit size and depth by construction.
To show that this circuit for $\MUX_{k+1}$ is hazard-free we make a case distinction.

If $s_{k+1}$ is stable, w.l.o.g.\ $s_{k+1}=0$, then $C$ outputs $\MUX_k(x_1,\ldots,x_{2^k};s_1,\ldots,s_k)$, since $\MUX_1$ is is hazard-free.
Thus if $\MUX_{k+1}$ has a hazard at $(x_1,\ldots,x_{2^{k+1}};s_1,\ldots,s_{k},0)$, then $\MUX_k$ has a hazard at $(x_1,\ldots,x_{2^k};s_1,\ldots,s_k)$.
But by induction hypothesis, $\MUX_k$ is hazard-free.

Now we consider the case $s_{k+1}=\mfu$.
For the sake of contradiction, assume that $\MUX_{k+1}$ has a hazard at $(x_1,\ldots,x_{2^{k+1}};s_1,\ldots,s_{k},\mfu)$.
Then all resolutions $(x_1',\ldots,x_{2^{k+1}}';s_1',\ldots,s_{k}',s_{k+1}') \in \IB^{2^{k+1}+k+1}$ of $(x_1,\ldots,x_{2^{k+1}};s_1,\ldots,s_{k},\mfu)$
yield $\MUX_{k+1}(x_1',\ldots,x_{2^{k+1}}';s_1',\ldots,s_{k}',s_{k+1}')= b$ for the same $b \in \IB$.
By construction of $C$ this implies $\MUX_{k}(x_1',\ldots,x_{2^{k}}';s_1',\ldots,s_{k}')= b = \MUX_{k}(x_{2^k+1}',\ldots,x_{2^{k+1}}';s_1',\ldots,s_{k}')$.
By induction hypothesis $\MUX_k$ is hazard-free. Thus we see
$\MUX_{k}(x_1,\ldots,x_{2^{k}};s_1,\ldots,s_{k}) = b = \MUX_{k}(x_{2^k+1},\ldots,x_{2^{k+1}};s_1,\ldots,s_{k})$.
This implies $\MUX_{k+1}(x_1,\ldots,x_{2^{k+1}};s_1,\ldots,s_{k},\mfu)=b$, because $\MUX_1$ is hazard-free.
This is a contradiction to $\MUX_{k+1}$ having a hazard at $(x_1,\ldots,x_{2^{k+1}};s_1,\ldots,s_{k},\mfu)$.

Putting both cases together we conclude that $\MUX_{k+1}$ is hazard-free.
\end{proof}
\begin{lemma}\label{lem:SpecComp}
Let $f:\IB^n \to \IB$ and $S\subseteq [n]$ with $|S|=k$. Then $L_{\{S\}}(f) < 2^k (L(f)+6)$ and $D_{\{S\}}(f) \le D(f)+4k$.
\end{lemma}
\begin{proof}
  For every assignment $\vec{a}\in \IB^{|S|}$, compute $g_{\vec{a}}=f(x|_{S\leftarrow \vec{a}})$,
  where $x|_{S\leftarrow \vec{a}}$ is the bit string obtained by replacing in $x$ the bits at the positions $S$ by the bit vector $\vec a$.
  We feed the results and the actual input bits from indices in $S$ into the hazard-free $k$-bit $\MUX$ from Lemma~\ref{lem:k_bit_mux} such that for stable values the correct output is determined.
  The correctness of the construction is now immediate from the fact that the $\MUX$ is hazard-free.

  Concerning the size bound, for each $\vec{a}\in \IB^{|S|}$ we have $L(g_{\vec{a}})\le L(f)$.
  Using the size bound for the $\MUX$ from Lemma~\ref{lem:k_bit_mux}, the construction thus has size smaller than $2^k(L(f)+6)$.
  Similarly, we combine $D(g_{\vec{a}})\le D(f)$ with the depth of the $\MUX$ to obtain the bound
  $D_{\{S\}}(f) \le D(f)+4k$.
\end{proof}

Using this construction as the base case, we seek to increase the number of sets (i.e., possible positions of the $k$ unstable bits) our circuits can handle.
\begin{theorem}\label{thm:kHazardFree}
  Let $T=\{S_1,\ldots,S_t\}$ be a collection of $k$-bit subsets of $[n]$.
  Then
  \begin{align*}
  L_T(f) \le \left(\frac{ne}{k}\right)^{2k} (L(f)+6) + O(t^{2.71}) \textit{~~~and~~~}
  D_T(f) \le D(f) + 8k + O(\log t).
\end{align*}
\end{theorem}
\begin{proof}
  We prove the theorem by a recursive construction with parameter $t$.
  \begin{itemize}
  \item Base case $t\le 2$:\\
  We use Lemma \ref{lem:SpecComp} to construct a circuit of size at most $2^{2k}(L(f)+6)$ and depth at most $D(f) + 8k$ that can cope with unstable inputs in $\bigcup_{S_i\in T}S_i$.
  \item Case $t\geq 3$:\\
  Set $T^A:=\{ S_1,\ldots S_{ \floor{\frac{2t}{3}} } \}$, $T^B:=\{ S_{ \floor{\frac{t}{3}}+1 },\ldots S_t \}$, and $T^C:=\{S_1,\ldots S_{ \floor{\frac{t}{3}} },S_{ \floor{\frac{2t}{3}}+1  },\ldots S_t \}$.
  Observe that each set $S_i\in T$ appears in exactly two sets among $T^A$, $T^B$, and $T^C$.
  We take circuits that yield correct outputs for each $S_i\in T^A$, $S_i\in T^B$, and $S_i\in T^C$, respectively, and feed their outputs into a hazard-free 3-input \emph{majority circuit,} i.e., a circuit returning the majority value of the three input bits.
  A hazard-free majority circuit of size $5$ and depth $3$ is given by taking the disjunction of the pairwise $\AND$s of the inputs (requiring two $\OR$ gates as we restrict to fan-in~$2$).
  At a potential $k$-bit hazard with unstable bits from $S_i\in T$, at least two of the inputs to the majority circuit are the same stable value.
  As we use a hazard-free majority circuit, this value is then the output (as the third value does not matter whether it is stable or not).
  \end{itemize}
It remains to examine the size and depth of the resulting circuit.
In each step of the recursion, the size of the considered collection of sets $T$ decreases by a factor of at least $1-\floor{t/3}/t$, where $t\geq 3$, hence the depth of the majority tree is $O(\log t)$. Solving the recursion with the Akra–Bazzi method~\cite{AB:98} yields $O(t^{\log_{3/2}3})$ as the size of the tree, where $\log_{3/2}3 < 2.71$. As the total number of possible leaves (base case) is at most $\binom{n}{2k}<\left(\frac{ne}{2k}\right)^{2k}$ and the size of each leaf is at most $2^{2k}(L(f)+6)$, the stated bound on the size of the circuit follows. The depth bound is obtained by adding the depth of leaves $D(f) + 8k$ to the depth of the tree.

Note that the above analysis does not take advantage of the repeated subproblems encountered by the recursive algorithm. By sharing subtrees that compute the same subproblem, it is possible to reduce the size of the majority tree and lower the exponent in $O(t^{2.71})$ (See Appendix~\ref{app:majtree}).
\end{proof}

Choosing $T=\binom{[n]}{k}$ and the bound $t = \binom{n}{k} < n^k$, we obtain Corollary~\ref{cor:kHazardFree}.

%% file: detection.tex
\section{Complexity of hazard detection}

In this section, we show that detecting hazards and detecting 1-bit hazards are both $\NP$-complete problems, see Theorem~\ref{thm:detecthazardnpcomplete} below.
The arguments are a bit subtle and thus we introduce several auxiliary hazard detection problems.

\begin{definition}
  We say that a circuit $C$ with $n$ inputs has a \emph{fixed hazard 
  at position $i\in [n]$} if $C$ has a $1$-bit hazard at a tuple $x\in \IT^n$ with $x_i=\mfu$.
\end{definition}

We fix some reasonable binary encoding of circuits and define the following languages:
\begin{itemize}
  \item $\fixedhazard=\{\langle C,i\rangle\,\mid\, C\mbox{ has a fixed hazard
    at position }i \}$
  \item $\SimpleHazard=\{C\,\mid\, C\mbox{ has a $1$-bit hazard}\}$
  \item $\hazard=\{C\,\mid\, C\mbox{ has a hazard}\}$
\end{itemize}

A circuit $C$ is called \emph{satisfiable} if there is a Boolean input for which $C$ outputs 1.
Otherwise $C$ is called \emph{unsatisfiable}.
We define a promise problem $\ZeroCircuitFixedHazard$:
Given an unsatisfiable circuit $C$ and $i\in [n]$, accept if $C$ has a
fixed hazard at position $i$.

\begin{lemma}\label{thm:fixedhazard}
   $\ZeroCircuitFixedHazard$ is $\NP$-hard.
\end{lemma}
\begin{proof}
  We reduce from circuit satisfiability as follows: To decide if a
  circuit $C$ on $n$ inputs is satisfiable, construct a circuit
  $C'=C\wedge (x_{n+1}\wedge \neg{x_{n+1}})$ where $x_{n+1}$ is a new variable.  Note
  that $C'$ is unsatisfiable by construction.  We claim that $C$ is
  satisfiable if and only if $C'$ has a fixed hazard at position $n+1$.
  \begin{itemize}
  \item [``$\Rightarrow$'':] Let $a$ be an assignment that satisfies $C$.
    Then $C'$ evaluates to $\mfu$ on input $(a, \mfu)$ and
    hence has a fixed hazard at position $n+1$.
  \item [``$\Leftarrow$'':] If $C$ is unsatisfiable,
    then $C'(a, y)=0$ for all $a \in \IB^n$, $y\in \IT$, and hence does not have a fixed hazard at position $n+1$.
    \qedhere
  \end{itemize}
\end{proof}

\begin{lemma}\label{lem:HazNPHard}
  The languages $\fixedhazard$, $\SimpleHazard$ and $\hazard$ are $\NP$-hard.
\end{lemma}
\begin{proof}
  Since $\ZeroCircuitFixedHazard$ is $\NP$-hard, the more general problem
  $\fixedhazard$ is also $\NP$-hard.

  We show that deciding the languages $\SimpleHazard$ and $\hazard$ is
  at least as hard as solving $\ZeroCircuitFixedHazard$.
  Let $C(x_1, \ldots, x_n)$ be an unsatisfiable circuit.
  Construct the circuit $C'=C(x_1,\ldots,x_n)\oplus x_2\oplus\cdots\oplus x_n$.
  We claim that $C'$ has a hazard if and only if $C$ has a fixed hazard at position $x_1$.
  \begin{itemize}
  \item [``$\Rightarrow$'':]
    Suppose $C'$ has a hazard. Note that since $C$ computes
    the constant $0$ function, $C'$ computes $x_2 \oplus \dots \oplus x_n$.
    If any of the input variables $x_2,\ldots,x_n$ has value $\mfu$,
    then $C'$ correctly outputs~$\mfu$.
    Thus, $C'$ can have a hazard only on inputs $a \in \IT^n$ that have exactly one $\mfu$
    occuring in the input position $1$. In this case $C(a)=\mfu$
    because otherwise $C'(a)$
    would be a Boolean value. Hence $C$ has a fixed hazard at
    position $1$.
  \item [``$\Leftarrow$'':] If $C$ has a fixed hazard at position $1$, then
    by definition, $C$ outputs $\mfu$ when $x_1=\mfu$ while all other
    inputs are stable. In this case, $C'$ also outputs $\mfu$ on this
    input. This is a hazard, since the Boolean function computed by $C'$ does
    not depend on $x_1$.
  \end{itemize}
  Thus, $\hazard$ is $\NP$-hard.

  Note that in the first part of this proof we actually proved that for
  the circuit $C'$ all hazards are $1$-bit hazards. So, the language
  $\SimpleHazard$ is also $\NP$-hard.
\end{proof}

\begin{lemma}\label{lem:HazNP}
  The languages $\fixedhazard$, $\SimpleHazard$ and $\hazard$ are in $\NP$.
\end{lemma}
\begin{proof}
  For $\fixedhazard$ and $\SimpleHazard$ we can take the input on which
  the circuit has a hazard as a witness. The verifier then has to check
  that the circuit actually outputs $\mfu$ on this input and that
  the outputs on the two stable inputs obtained by replacing $\mfu$ by $0$
  and $1$ match.

  For $\hazard$, the verifier cannot check the definition directly, since
  the number of resolutions can be exponential. However, if a circuit
  $C$ has a hazard, then there exists an input $x \in \IT^n$ with $C(x)=\mfu$
  such that on the inputs $x^{(0)}$ and $x^{(1)}$ that are obtained from $x$ by replacing
  the \emph{leftmost} $\mfu$ by $0$ and $1$ respectively the circuit $C$
  outputs the same stable value $b$.
  This can be seen as follows.
  Let $H_C \subset \IT^n$ be the set
  of all inputs on which $C$ has a hazard. Any element $x$ that is maximal in $H_C$
  with respect to $\preceq$ satisfies the requirement: since $x$ is a hazard,
  $C(x)=\mfu$ and
  the output of $C$ on all resolutions of $x$ is the same stable value $b$.
  Thus the output of $C$ on all resolutions of $x^{(0)}$ and of $x^{(1)}$ is $b$.
  Since $x$ is maximal, both $x^{(0)}$ and $x^{(1)}$ do not lie in $H_C$, which implies $C(x^{(0)})=C(x^{(1)})=b$.
  Such $x$ with $C(x)=\mfu$ and $C(x^{(0)})=C(x^{(1)})=b$ can be used as a witness for $\hazard$.
\end{proof}

From Lemma \ref{lem:HazNPHard} and Lemma \ref{lem:HazNP}, we conclude:
\begin{theorem}\label{thm:detecthazardnpcomplete}
  The languages $\fixedhazard$, $\SimpleHazard$ and $\hazard$ are $\NP$-complete.
\end{theorem}

%% file: future.tex
\section{Future directions}\label{sec:future}
Hazard-free complexity is an interesting subject that arises in practice when constructing physical circuits.
Theorem~\ref{thm:monmetaeq} sends the strong message that hazard-free complexity is of interest even \emph{without} its application in mind.
As a theoretical model of computation, hazard-free circuits are at least as interesting as monotone circuits.
Section~\ref{sec:hazard-free} hints towards the fact that there is a rich intermediate landscape of $k$-hazard free circuits to be analyzed for different ranges of $k$.
This can potentially be very illuminating for our understanding of the nature and limits of efficient computation in general.

Given the lower bound corollaries to Theorem~\ref{thm:monmetaeq} and the circuit construction in Corollary~\ref{cor:kHazardFree},
one dangling open question is the fixed parameter tractability of $k$-hazard-free circuits:
Does there exist a function $\varphi$ such that for all sequences of Boolean functions $f_n:\IB^n\to\IB$ there exist $k$-hazard-free circuits of size $\varphi(k)\cdot \poly(n)$?

A further direction of interest is to understand the power of masking registers~\cite{FFL:16}, both in terms of computational power and efficiency.
It is neither known precisely which functions can be computed by a clocked circuits within $r$ rounds, and it is not clear whether a factor $\Omega(k)$ overhead for computing the closure with masking registers is necessary.

%% file: bases.tex
\section{Circuits with different basic gates}

In the main part of this paper we used circuits with $\AND$-, $\OR$- and $\NOT$-gates, because this is one of the standard models in circuit complexity.
But we have also already seen that the circuit transformations we use for proving lower bounds rely only on the general construction of the derivative and can be performed on circuits with arbitrary gates, not just $\AND$-, $\OR$-, and $\NOT$-gates.
In this appendix we show that any other functionally complete set (in the sense of e.g.\ \cite{End:01}) can be used to give an equivalent theory of hazard-free complexity and natural functions, see the upcoming Corollary~\ref{cor:functionallycomplete}.
A priori it is not obvious that every function can be implemented by a hazard-free circuit over some set of gates, even if the set of gates is functionally complete in the Boolean sense.
We prove that everything works properly if we allow constant input gates.
This subtlety is unavoidable, since any nontrivial natural function outputs $\mfu$ if all inputs are $\mfu$, so any circuit without constant gates also has this property.
Therefore, the constant function is not computable by hazard-free circuits without the use of constant gates.

\cite[Theorem 2]{Brz:97} shows that every natural function can be implemented over the set of functions $\Phi = \{\AND, \OR, \NOT, 1\}$.
Using the fact that a hazard-free implementation of $\OR$ can be achieved via the standard De Morgan implementation $x \ \OR\ y = \NOT((\NOT\ x) \ \AND\ (\NOT\ y))$,
it follows that
\begin{equation}\label{eq:andnot-natural}
\text{every natural function can be implemented over the set of functions $\{\AND, \NOT, 1\}$.}
\end{equation}

A Boolean function $f \colon \IB^n \to \IB$ is called \emph{linear} if there exist $a_0,\ldots,a_n \in \IB$ with $f(x_1,\ldots,x_n)=a_0 \oplus a_1 x_1 \oplus \dots \oplus a_n x_n$. Otherwise $f$ is called \emph{nonlinear}.
The composition of linear functions is linear, but not all Boolean functions are linear.
Thus every functionally complete set must contain a nonlinear function.

Variants of the following lemma are often used as a part of proof of Post's theorem characterizing functionally complete systems.
\begin{lemma}\label{lem:substitutionform}
  Let $f \colon \IB^n \to \IB$ be a nonlinear Boolean function.
  Then $n \geq 2$.
  Moreover, by substituting constants for some input variables of $f$, we can obtain a function of $2$ variables of the form $(x_1 \oplus c_1)(x_2 \oplus c_2) \oplus c_0$.
\end{lemma}
\begin{proof}
Using the fact that over the field $\mathbb{F}_2$ with two elements we have $x_i \ \AND \ x_j = x_i\cdot x_j$ and $\NOT\ x_i = x_i \oplus 1$, we can represent $f$ as a polynomial over $\mathbb{F}_2$. Using that $(x_i)^k=x_i$ for $k \geq 1$, we can
  represent $f$ in its algebraic normal form
  \[f(x_1, \dots, x_n) = \bigoplus_{I \subset [n]} a_I \prod_{i \in I} x_i,\]
  where each $a_I \in \IB$.
  We call $I$ a \emph{monomial} and call $|I|$ its \emph{degree}.
  Since $f$ is nonlinear, there is at least one monomial of degree at least 2 with nonzero coefficient $a_I$.
  Thus we proved $n\geq 2$.
  Among monomials of degree at least $2$, choose one monomial of minimal degree and set all the variables not contained in this monomial to~$0$.
  Without loss of generality, the chosen monomial is $x_1 \cdots x_t$, $t \geq 2$.
  The resulting function has the form
  \[x_1 \dots x_t \oplus a_1 x_1 \oplus \dots \oplus a_t x_t \oplus a_0.\]
  Setting all variables except $x_1$ and $x_2$ to $1$, we obtain $x_1 x_2 \oplus a_1 x_1 \oplus a_2 x_2 \oplus a_0'$, or $(x_1 \oplus c_1)(x_2 \oplus c_2) \oplus c_0$ where $c_1 = a_2$, $c_2 = a_1$ and $c_0 = a_0' \oplus a_1 a_2$.
\end{proof}

\begin{theorem}\label{thm:functionallycomplete}
  Let $\Phi$ be a set of natural functions such that their restrictions to $\IB$ form a functionally complete set.
  Suppose $\Phi$ contains an extension of a nonlinear Boolean function that is free of $1$-bit hazards.
  Then every natural function can be computed by a circuit over $\Phi$ using the constant~$1$.
\end{theorem}
\begin{proof}
  In the light of \eqref{eq:andnot-natural}, it is enough to show that hazard-free circuits for $\NOT$ and $\AND$ can be implemented over $\Phi$.
  The statement is trivial for the negation: since $\NOT$ has only one natural extension, any circuit that computes it is automatically hazard-free.
  Using $\NOT$, we can obtain the constant $0$ from the constant $1$.

  By Lemma~\ref{lem:substitutionform}, we obtain from the $1$-hazard-free nonlinear function contained in $\Phi$ a function of the form $(x_1 \oplus c_1)(x_2 \oplus c_2) \oplus c_0$ by substituting constants $0$ and $1$ into this nonlinear function.
  Constant substitution does not introduce hazards. Since $\NOT\ x = x \oplus 1$, we can transform the circuit $C$ computing $(x_1 \oplus c_1)(x_2 \oplus c_2) \oplus c_0$ to a circuit $C'$ computing $x_1 x_2$ by placing $\NOT$ on input $x_i$ if $c_i = 1$ and on the output if $c_0 = 1$.
In other words, $C'(x_1,x_2) = C(x_1 \oplus c_1,x_2 \oplus c_2)\oplus c_0$.
  
  Let us check that $C'$ is hazard-free.
  The circuit $C'$ is computing the conjunction and thus can have hazards only on two inputs: $(0, \mfu)$ and $(\mfu, 0)$.
  If $C'(0, \mfu)=\mfu$, then $C(c_1, \mfu)=\mfu$.
  This is a $1$-bit hazard, since $(c_1 \oplus c_1)(x \oplus c_2) \oplus c_0 = c_0$ for all $x \in \IB$.
  The other case is analogous.
\end{proof}

\begin{corollary}\label{cor:functionallycomplete}
Given a functionally complete set of Boolean functions, let $\Phi$ be the set of their hazard-free extensions.
Every natural function can be computed by a circuit over $\Phi$ using the constant~$1$.
\end{corollary}
\begin{proof}
Since a functionally complete set cannot only consist of linear functions, at least one function must be nonlinear.
A hazard-free function in particular does not have a 1-hazard.
Thus Theorem~\ref{thm:functionallycomplete} applies.
\end{proof}

%% file: majtree.tex
\section{Improved upper-bounds for $k$-bit hazard-free circuits}\label{app:majtree}

In this section, we show how to improve the analysis in the proof of
Theorem~\ref{thm:kHazardFree} to show an improved upper-bound of
$O(t^{2.46})$ on the number of majority gates in the circuit. This
yields the bound
$L_T(f) \le \left(\frac{ne}{k}\right)^{2k} (L(f)+6) + O(t^{2.46})$ on
the overall size of the circuit. The depth of the circuit is
unchanged.

The majority tree in the proof of Theorem~\ref{thm:kHazardFree} is
constructed by dividing the problem denoted by the set $T$ of size $t$
into three subproblems denoted by subsets $T^A$, $T^B$, and $T^C$ each
of size $2t/3$.
The improved
upper-bound is obtained by analysing these recursive subproblems at
larger depths and observing that many of them coincide. For example,
the number of distinct recursive subproblems at depth three is only
$25$ instead of $27$ because $((T^A)^A)^A = ((T^C)^A)^A$ (Both these sets are
the first $8t/27$ elements of $T$) and $((T^B)^B)^B = ((T^C)^B)^B$ (Both these
sets are the last $8t/27$ elements of $T$). This yields the recurrence
$F(t) \leq 25F(8t/27) + O(1) = O(t^{2.65})$ on the number of
majority gates. A similar analysis of the recursive subproblems using
a computer program shows that there are only $410040$ distinct
recursive subproblems at depth 13 yielding the recurrence
$F(t) \leq 410040 F({(2/3)}^{13} t) + O(1) = O(t^{2.46})$ on the
number of majority gates. The improved upper-bound follows.

After this paper got accepted, a new construction of circuits that are free of
$k$-bit hazards was discovered that uses only
$\left(\frac{ne}{k}\right)^{2k} (L(f)+6) + O(t^{2})$ gates, thus improving
further on Corollary~\ref{cor:kHazardFree}. This new construction is
sufficiently different from the one used in Theorem~\ref{thm:kHazardFree} and we
plan to develop it further in an extended version of this work.